\documentclass[sigconf,nonacm,balance=false]{acmart}
\usepackage{xspace}
\usepackage{enumitem}
\usepackage{xcolor,framed}
\usepackage{algpseudocode}
\usepackage[T1]{fontenc}
\usepackage[linesnumbered,ruled,vlined,noend]{algorithm2e}
\NoCaptionOfAlgo

\AtBeginDocument{%
  }

\newcommand{\CollCost}{\mathcal{C}\xspace}
\newcommand{\defn}[1]{\textbf{\emph{#1}}}
\newcommand{\poly}[1]{{\texttt{poly}({#1})}}
\newcommand{\mainAlgorithm}{\textsc{Aim-High}\xspace}
\newcommand{\mainAlgorithmAcrynom}{\textsc{AH}\xspace}
\newcommand{\wakeup}{{wakeup}\xspace}
\newcommand{\whp}{w.h.p.\xspace}
\newcommand{\Whp}{W.h.p.\xspace}
\newcommand{\tunableParam}{\epsilon\xspace}
\newcommand{\currWin}{w_{\mbox{\tiny cur}}\xspace}
\newcommand{\con}{\texttt{Con}(t)\xspace}
\newcommand{\conSquare}{\texttt{Con}^2(t)\xspace}

\begin{document}

\title{A Gentle Wakeup Call\\ Symmetry Breaking with Less Collision Cost}
\renewcommand{\shorttitle}{A Gentle Wakeup Call}

\author{Umesh Biswas}
\orcid{0009-0000-2842-5527}
\email{ucb5@msstate.edu}
\affiliation{%
\department{Department of Computer Science and Engineering}
  \institution{Mississippi State University}
  \city{Mississippi State}
  \state{Mississippi}
  \country{USA}
}

\author{Maxwell Young}
\orcid{0000-0002-5251-8595}
\email{myoung@cse.msstate.edu}
\affiliation{%
\department{Department of Computer Science and Engineering}
  \institution{Mississippi State University}
  \city{Mississippi State}
  \state{Mississippi}
  \country{USA}
  \vspace{1cm}
}

\begin{abstract}
The wakeup problem addresses the fundamental challenge of symmetry breaking. Initially,  $n$ devices share a time-slotted multiple access channel, which models wireless communication. A transmission succeeds if exactly one device sends in a slot; if two or more transmit, a collision occurs and none succeed. The goal is to achieve a single successful transmission efficiently.

Prior work primarily analyzes latency---the number of slots until the first success. However, in many modern systems, each collision incurs a nontrivial delay, $\CollCost$, which prior analyses neglect. Consequently, although existing algorithms achieve polylogarithmic-in-$n$ latency, they still suffer a delay of $\Omega(\CollCost)$ due to collisions.

Here, we design and analyze a randomized wakeup algorithm, \mainAlgorithm. When $\CollCost$ is sufficiently large with respect to $n$, \mainAlgorithm has expected latency and expected total cost of collisions that are nearly $O(\sqrt{\CollCost})$; otherwise,  both quantities are $O(\poly{\log n})$. Finally, for a well-studied class of algorithms, we establish a trade-off between latency and expected total cost of collisions.
\end{abstract}

\maketitle

\section{Introduction} 

In the wakeup problem, devices are either \defn{active} ({\it awake}) or \defn{inactive} ({\it asleep}). Each active device has a packet to send on a  shared channel, while inactive devices are limited to listening. The first successful transmission serves as a “wakeup call”, awakening all inactive devices and resolving the problem \cite{banicescu2024survey}.

Prior work on this problem has focused on the time until the first success, which is typically referred to as the \defn{latency}. Minimizing this metric is challenging, since in any unit of time, the shared channel can only accommodate a single sender; if two or more devices send simultaneously, the result is a \defn{collision} where none of the corresponding packets are successfully sent. Sending on the channel is performed in a distributed fashion; that is, {\it a priori} there is no central scheduler or coordinator. 

Under the standard cost model for wakeup, each collision wastes a single slot. However, in many settings, a collision {\it can consume much  more than a single slot}. Therefore, collisions may significantly impact the performance of modern systems.

In this work, we address both latency and the cost of collisions required to solve the wakeup problem.  While  existing algorithms give strong latency guarantees, they result in a large number of collisions (see  Section \ref{sec:tech-overview}), and  addressing these two metrics in tandem requires new ideas.\medskip

\noindent{\bf Why Should We Care About Collisions?} The delay incurred by a collision can vary across different systems. As an example, consider WiFi (IEEE 802.11) networks; here, both sending a packet and the delay from a collision consume time, say some large value $\CollCost$, proportional to the packet size \cite{anderton:windowed}. In this setting, wakeup incurs an unavoidable delay of $\CollCost$, since a single success is required. Yet, reducing the delay caused by collisions is still important. If this delay is, say, $X$, then $\CollCost/(\CollCost + X)$ provides a measure of how much time is wasted due to collisions. So, if $X=\Theta(\CollCost)$, then this ratio is suboptimal by a constant factor. However, if $X=\Theta(\sqrt{\CollCost})$, then $\CollCost/(\CollCost + \Theta(\sqrt{\CollCost})) \geq 1 - \Theta(1/\sqrt{\CollCost})$; that is, the ratio approaches the optimal value of $1$ as $\CollCost$ grows large.

More generally, over long-range links, where the propagation delays are large, a collision means devices can wait a long round-trip time before realizing the transmission has failed. Also, in intermittently-connected mobile wireless networks \cite{zhang2006routing}, failed communication due to a collision may result in the sender having to wait until the intended receiver is again within transmission range. Finally, in environments where there is some fluctuation in the delay (e.g., arising from fluctuating channel conditions or periods of heavy network usage), $\CollCost$ can be interpreted as an upper bound on this delay that holds under most operating conditions.

\subsection{Model and Notation}\label{sec:model}

We consider a system with a large, {\underline{unknown}} number of {\boldmath{$n$}} devices initially present, each with a single packet to send on a shared channel. Going forward, for ease of presentation, we refer only to the packets---rather than devices---committing actions such as sending and listening. \medskip

\noindent{\bf Communication.} Time is divided into disjoint \defn{slots}, each of which can accommodate a packet transmission. Communication occurs on a \defn{multiple access channel}, which is defined as follows.  In any slot, a packet can either attempt to transmit or listen to the channel. If no packet is transmitted in a given slot, it is \defn{empty}. If only one packet is transmitted, the packet \defn{succeeds}; we often refer to this as a \defn{success}. However, if multiple packets transmit simultaneously, they all fail---this is a \defn{collision}---and they may try again later. There is no scheduler or central authority {\it a priori}.\smallskip

\noindent{\bf No Collision Detection.} There is no mechanism to distinguish between an empty slot and a collision—that is, {\it no collision detection} (no-CD). In many practical systems, collision detection (CD) is unavailable or unreliable. For example, in WiFi networks, received signal strength may offer a weak form of CD, but it is error-prone. Consequently, many prior results address the challenging no-CD setting (see the survey \cite{banicescu2024survey} and related discussion in Section \ref{sec:related-work}). 

We emphasize that CD and the cost of collisions are orthogonal concerns: collisions incur a performance penalty regardless of whether they are detectable. For example, in WiFi, even if devices cannot sense an ongoing collision, the channel remains unavailable for its duration.
\medskip

\noindent{\bf Metrics.} We care about   \defn{latency}, which refers to the time between the first packet injection and the first success.
Additionally, each collision incurs a {
\underline{known}} delay of {\boldmath{$\CollCost$}} $\geq 4$.\footnote{The value $4$ is chosen to simplify the analysis. We also note that only $\CollCost = \omega(1)$ is interesting;  otherwise, an optimal algorithm (such as sawtooth backoff \cite{GreenbergL85,Gereb-GrausT92,banicescu2024survey}) for the standard model will be asymptotically optimal.}  To capture the worst-case performance,  $\CollCost$ is chosen by the adversary with knowledge of $n$,  subject to $\CollCost = O(\poly{n})$, where the degree of the polynomial is unknown. Note that this reveals no information about $n$, whose value is unknown to the packets {\it a priori}.

Given any \wakeup algorithm, the two pertinent costs are: (i) the \defn{total cost of collisions}, which is the number of collisions incurred over the execution of the algorithm multiplied by $\CollCost$, and (ii) the latency. We report on both of these metrics separately, with the aim of minimizing their maximum.
\medskip

\noindent{\bf Notation.} An event holds \defn{with high probability (\defn{\whp})} in $x$ if it occurs with probability at least $1-O(1/\poly{x})$, where we can make the constant degree of the polynomial  arbitrarily large. We use $\log^c y$ to denote $(\log y)^c$.


\subsection{Our Results}

The performance of our algorithm, \defn{\mainAlgorithm (\mainAlgorithmAcrynom)}, is parameterized by {\boldmath{$\tunableParam$}}, which is a small positive constant set by the user prior to execution. 

\begin{theorem}\label{thm:static_upper}
\mainAlgorithmAcrynom solves the wakeup problem with the following costs:
\begin{enumerate}[leftmargin=15pt]
     \item If $\CollCost = \Omega(\log^{\Theta(1/\tunableParam)}n)$, then the expected latency is $O(\CollCost^{\frac{1}{2}+2\tunableParam})$ and  the expected total cost of collisions is  $O(\CollCost^{\frac{1}{2}+\tunableParam})$.
     \item Else, the expected latency is $O(\log^{5+\frac{1}{2\tunableParam}} n)$  and the expected total cost of collisions is $O( \log^{5 + \frac{3}{2\tunableParam}}n)$.
\end{enumerate}
\end{theorem}

By comparison, existing approaches to wakeup incur a collision cost that is $\Omega(\CollCost)$ (see Section \ref{sec:tech-overview}). In contrast, Theorem \ref{thm:static_upper} implies that if $\CollCost$ is ``large'', \mainAlgorithmAcrynom offers performance that is  sublinear in $\CollCost$. Else,  latency and  expected total cost of collisions are  polylogarithmic-in-$n$, which echoes prior wakeup results (see Section \ref{sec:related-work}). 

Many upper bounds are achieved by ``fair'' algorithms. A \defn{fair algorithm} is defined by a sequence $p_1$, $p_2$, $\dots$, where $p_i$ denotes the sending probability used by every packet in the $i$-th slot. In Section \ref{sec:lower-bound-static},  we establish the following trade-off between latency and expected total cost of collisions. 

\begin{theorem}\label{thm:static_lower}
Let $\alpha$ and $\beta$ be constants such that $1/(2\tunableParam)\geq \alpha \geq 0$ and $\beta\geq 0$. Any fair algorithm for the wakeup problem that  has $O(\CollCost^{(1/2)+\alpha\tunableParam} \log^{\beta}n)$ latency has an expected total cost of collisions $\Omega(\CollCost^{(1/2)-\alpha\tunableParam}/\log^{\beta}n)$. 
\end{theorem}


\section{Related Work}\label{sec:related-work}

Wakeup, or leader election, addresses the fundamental problem of symmetry breaking. Early work by Willard~\cite{willard:loglog} establishes a $\lg\lg n + O(1)$ upper bound with a matching lower bound for fair algorithms under collision detection (CD).

In contrast,  {\it without CD, a common theme is that the latency increases to logarithmic or polylogarithmic in $n$}. For starters,  without collision detection and when  $n$ is known, Kushilevitz and Mansour \cite{kushilevitz1998omega,kushilevitz1993omega} prove a lower bound of $\Omega(\log n)$ expected latency.  Newport \cite{newport2014radio,newport:radio-journal}, in work that unifies many prior lower-bound results, shows $\Omega(\log n)$ expected latency and $\Omega(\log^2 n)$ latency \whp for general randomized algorithms. Complementing these lower bounds, Bar-Yehuda et al. \cite{bar1992time} achieve a matching upper bound of $O(\log n)$ expected latency and  \whp $O(\log^2 n)$ latency.

Farach-Colton et al. \cite{farach-colton:lower-bounds} obtain $\Omega(\log n \log(1/\varepsilon))$ latency for algorithms without CD, succeeding with probability at least $1-\varepsilon$. A complementary  by  Jurdziński and Stachowiak \cite{jurdzinski2002probabilistic}  demonstrates an algorithm that succeeds with probability at least $1-\varepsilon$ with $O(\log n \log(1/\varepsilon))$ latency, without CD, but assuming unique labels or knowledge of $n$.

With multiple channels and collision detection, Fineman et al. \cite{fineman:contention2} give an efficient contention-resolution algorithm. In the SINR model \cite{moscibroda2006complexity}, Fineman et al. \cite{fineman:contention}  achieve logarithmic latency and establish matching lower bounds. Chang et al. \cite{chang:exponential-jacm} examine energy efficiency under varying feedback models, while Gilbert et al. \cite{gilbert2021contention} study the effect of system-size prediction. 

{\it Research on the cost of collisions is still in its early stages.}  Anderton et al. \cite{anderton:windowed} show that collisions create substantial discrepancies between theoretical and observed performance in backoff algorithms, underscoring the need to account for collision costs. Building on this, Biswas et al. \cite{biswas2024softening} provide one of the first analyses, addressing the problem of contention resolution and deriving bounds on expected latency and total expected collision cost.

\begin{algorithm}[t!]
\caption{\bf \mainAlgorithm}\label{our-dynamic-algoritm-LE}

\SetKwFunction{FMain}{Elected-Leader}
\SetKwFunction{FGoodCon}{Good-Contention}
\SetKwProg{Fn}{Function}{:}{}

\BlankLine
\tcp*[h]{Initialization}\\
$w_{0} \leftarrow 2^{\CollCost^{\tunableParam}}$ \label{alg:first-line-dynamic}\\
$ j \leftarrow 1$\label{alg:third-line-dynamic}\\
\medskip
\medskip
\tcp*[h]{Iterations}\\
\While{true}{\medskip  \label{alg:iteration}
\tcp*[h]{Halving Phase}\\
$ \currWin \leftarrow w_{0}$ \\ 
\While{$\currWin \geq 2$}{{\label{alg:dynamic-first-while-loop}}
    \For{slot $i = 1$ \KwTo $d\sqrt{\CollCost}\,\ln(\currWin)$}{\label{alg:dynamic-first-while-for-loop}}{
        \vspace{-2pt}Send packet with probability $1/w_{cur}${\label{alg:dynamic-second-sending-probability}}\\
        \If{success}{Terminate}
        \Else{$\currWin\leftarrow \currWin/2$\label{alg:dynamic-havle-assign}\\}
    }}
\medskip\medskip
\tcp*[h]{Doubling Phase}\\
$ \currWin \leftarrow w_{0}$\\ 
\For{$s=1$ to $j$}{\label{alg:while_second_dynamic}
        $\currWin\leftarrow 2\currWin $\label{alg:dynamic-double-assign}\\
        \For{slot $i = 1$ \KwTo $d\sqrt{\CollCost}\,\ln(\currWin)$}{\label{alg:dynamic-second-while-for-loop}}{
        \vspace{-2pt}Send packet with probability $1/w_{cur}${\label{alg:dynamic-third-sending-probability}}\\
        \If{success}{Terminate}
     }
   } 
   $j \leftarrow j+1$
  } 
\end{algorithm}

\section{Algorithm Specification}\label{sec:our_algo_static}

The pseudocode for \mainAlgorithmAcrynom is given above, and each packet executes this algorithm upon being activated. Under \mainAlgorithmAcrynom, each packet begins with an initial window size of {\boldmath{$w_0$}} $=2^{\CollCost^{\tunableParam}}$. 

Let {\boldmath{$\currWin$}} denote the current window size. Initially, the  $\currWin$ is set to $w_0$, and this is updated throughout the execution, which proceeds by alternating a \defn{halving phase} and a \defn{doubling phase}, described in detail momentarily.  Each execution of a halving phase followed by the corresponding doubling phase is an \defn{iteration}. The parameter {\boldmath{$j$}} in our pseudocode indexes the iterations.
\medskip

\noindent{\underline{\it Halving phase}.} Every packet sends with probability  $\Theta(1/\currWin)$ in each of $d\sqrt{\CollCost}\,\ln(\currWin)$ slots, where {\boldmath{$d$}} $>0$ is a constant (Lines \ref{alg:dynamic-first-while-for-loop} and \ref{alg:dynamic-second-sending-probability}); we refer to these slots as a \defn{sample}. The constant $d$ is key to the bounds on error probability in our results (specifically, the \whp bound in Lemma \ref{lem:success-going-down}).

After finishing a sample, if no success has occurred, then  each packet sets $\currWin\leftarrow \currWin/2$. This is repeated until either a success (and termination) occurs, or further halving would reduce the  window size below $2$ (Line \ref{alg:dynamic-first-while-loop}). At which point, the doubling phase begins.\medskip

\noindent{\underline{\it Doubling phase}.} At the beginning of this phase,  $\currWin$ starts at $2w_0$. Each packet sends with probability  $\Theta(1/\currWin)$ in each slot of a sample (Line \ref{alg:dynamic-second-while-for-loop} and \ref{alg:dynamic-third-sending-probability}). If a success occurs, the execution terminates; otherwise, the window doubles and another sample is executed. This process continues until $j$ samples (each corresponding to a doubling of the window) are performed, where $j=1$ initially, and increments at the end of each iteration. \smallskip

\subsection{Context and Overview}\label{sec:tech-overview}

We start by considering how prior approaches fare in a model when there is a collision cost. Then, we discuss our design choices regarding \mainAlgorithmAcrynom, along with an overview of our analysis.

\subsubsection{\bf A Baseline via Backoff.} A natural approach for breaking symmetry is to ``backoff'', where each packet sends with probability $1/2^i$ in slot $i \geq 0$. At $i = \lg(n)$, the probability of a success in a slot becomes constant. An example of its use in solving wakeup is given by Jurdzi{\'n}ski and Stachowiak \cite{jurdzinski2002probabilistic}, who achieve $O(n/\log n)$ latency  with bounded error $\varepsilon$ in a similar model to ours (without a  notion of collision cost); however, each active packet starts with a high (i.e., constant) sending probability, which  inflicts $\Omega(\CollCost)$ cost.

So, we can observe that while backoff may be an important part of any solution, a key obstacle is how to set the initial sending probability suitably low. Since $n$ is unknown {\it a priori}, it seems that no matter what initial sending probability we choose---even if it depends on the known value of $\CollCost$---the adversary may choose a sufficiently large value of $n$ such that collisions are unavoidable. What can we do?

\subsubsection{\bf Design Choices and Overview.} \label{sec:design-choices}We now discuss how the design of \mainAlgorithmAcrynom allows us to do better than prior results in a model with collision cost.\medskip

\noindent{\textbf{Setting an Initial Window Size.}} The initial sending probability needs to be set carefully. Too high and we incur many  collisions as we discussed above. 

Consider a window of size $w^* = \Theta(n\sqrt{\CollCost})$; we refer to this as the size of a \defn{good window}.  Informally, once packets reach a good window, the chance of a success in any slot of the window is at least $\binom{n}{1}(1/w^*)\approx 1/\sqrt{\CollCost}$. Therefore, over $\Theta(\sqrt{\CollCost})$ slots, we expect a success. At the same time, the probability of a collision is approximately $\binom{n}{2} (1/w^*)^2 \approx 1/\CollCost$, and thus the expected collision cost per slot is $1$, and $O(\sqrt{\CollCost})$ over $\Theta(\sqrt{\CollCost})$ slots. This observation dictates the size of our samples. 

In designing \mainAlgorithmAcrynom, we hope that $w^* \leq w_0$, since in this case, the above reasoning yields good expected collision cost (and, as we will show, good latency). However, since the adversary sets $\CollCost$ with knowledge of $n$, our hopes may be dashed; the adversary can always make collisions overwhelmingly likely. However, in this case, our consolation prize is that collisions are cheap, and the degree of cheapness corresponds to $w_0$. We choose to "{\it aim high}" with the initial window size; that is, setting $w_0 = 2^{\CollCost^{\tunableParam}}$. This choice sets up our analysis of the following {\bf two cases}: 

\begin{enumerate}[leftmargin=14pt]
\item{\bf If {\boldmath{$n\sqrt{\CollCost}\leq w_0$}}}, then the good-window size is less than or equal to the initial-window size. Here, the probability of a collision is not too high. Thus, so long as we adjust the sending probability upwards  in a conservative manner, we can reduce expected collision cost; this is the purpose of the halving phase. Of course, there is a tension with latency, since this adjustment cannot be {\it too} conservative, otherwise the latency will be too large.\smallskip  

\item{\bf Else {\boldmath{$n\sqrt{\CollCost}> w_0$}}}, and the good-window size is larger than the initial-window size. Here, the collision probability may be ``too high''. However, a key realization is that $\CollCost$ is small;  we will argue that $n> 2^{\CollCost^{\tunableParam}}/\sqrt{\CollCost}$ implies that $\CollCost = O(\log^{1/\epsilon}n)$. Thus, we can tolerate these collisions as we increase the window size in the doubling phase. 
\end{enumerate}
\smallskip

\noindent{\bf Why Iterations?} Although \mainAlgorithmAcrynom is likely to terminate when a good window is reached, there is some small probability of failure. Suppose we are in Case 1, but this failure occurs in the halving phase. Then, this small probability of failure yields good expected total cost of  collisions as the window size decreases past the good window down to $2$ (see Lemma \ref{lem:down_coll_cost}). However, if the halving phase is never revisited, the packets may accrue intolerable latency by executing the doubling phase without success (as the sending probability drops off exponentially fast). Conversely, suppose we are in Case 2 and this failure occurs. Then, if we never reset $\currWin$ to (roughly) $w_0$, the latency costs may again be high. Therefore, iterations facilitate good expected latency (see Lemma \ref{lem:upward_latency}).
\medskip

\noindent{\bf Parameter {\boldmath{$\tunableParam$}}.} Throughout, we treat $\tunableParam$ as a small constant for the purposes of our analysis. Yet, it is worth pointing out how it interacts with the latency and the expected total cost of collisions. The former is improved by making $\epsilon$ small, since this shortens the halving phase. So, if aiming high pays off, then a smaller $\tunableParam$ can reduce latency.  However, we pay for this in terms of the latter metric; specifically, when $n\sqrt{\CollCost}>w_0$ and the algorithm looks to succeed in the doubling phase, we have $\CollCost = O(\log^{1/\tunableParam}n)$. Thus, the degree of the polynomial increases with the reciprocal of $\tunableParam$.


\section{Upper Bound}\label{sec:upper-bound-static-setting}

In this section, we analyze the performance of \mainAlgorithmAcrynom; due to space constraints, we omit several proofs. Define {\boldmath{$\eta$}} $= \max\{n, \CollCost\}$. A good window was loosely discussed in Section \ref{sec:design-choices}; we define it as any window with size in $[n\sqrt{\mathcal{C}}, 3 n\sqrt{\mathcal{C}}]$. The next lemma argues that success is likely in a good window.

\begin{lemma}\label{lem:success-going-down}
\Whp in $\eta$, a success occurs if \mainAlgorithmAcrynom executes in a good window.  
\end{lemma}

\subsection{{\bf Case 1:} {\boldmath{$n\sqrt{\CollCost} \leq w_0$.}}}\label{sec:downward_direction}
In this case,  the good window size is less than or equal to $w_0$, so we argue that \mainAlgorithmAcrynom is likely to obtain a success during the halving phase. We make use of the following known result that upper bounds the probability of a collision.

\begin{lemma}\label{lem:upper-prob-coll}(Biswas et al. \cite{biswas2024softening})
Consider any slot $t$ where there exist $m$ active packets, each sending with the same probability $p$, where  $p< 1/m$. The probability of a collision in slot $t$ is at most $\frac{2m^2p^2}{\left(1-mp\right)}$.
\end{lemma}

\begin{lemma}\label{lem:down_coll_cost}
Assume that $n\sqrt{\CollCost} \leq w_0$. The expected  total cost of collisions  for \mainAlgorithmAcrynom is $O(\CollCost^{(1/2)+\tunableParam})$.
\end{lemma}
\begin{proof}
For any fixed iteration $j\geq 1$ (i.e., an execution of the while loop of Line \ref{alg:iteration} in our pseudocode), we can analyze the expected total collision cost as follows. Let {\boldmath{$q$}} denote the probability that no success occurs 
in the halving phase of iteration $j$; note that $q$ is the same for each iteration. We accrue collisions that occur in the entire halving phase and the following doubling phase of iteration $j$. Let {\boldmath{$B_j$}} denote the cost of those collisions that occur in any slot from the start of iteration $1$ to until the end of iteration $j$. 

Additionally, let {\boldmath{$G$}} denote the contribution to the expected cost for the first halving phase of the first iteration. 
Over all iterations, the expected total collision cost is thus at most:
\begin{align}
    &  G + \sum_{j=1}^{\infty} q^{j} \cdot B_j \label{eq:expectation}
\end{align}

{\it {\underline{Upper Bounding $B_j$}}.} Clearly, $B_j$ is upper-bounded by the cost of having a collision in every single slot of iterations $1$ through $j$. In any iteration $k$, there is a sample of $d\sqrt{\CollCost}\ln(\currWin)$ slots for each of the at most $\lg(w_0)$ windows; thus, over iterations $1$ through $j$, the halving phases contribute  $O(j\sqrt{\CollCost}\ln(\currWin) \lg(w_0))$ slots. Additionally, in iteration $k$, the corresponding doubling phase executes $k$ samples, each consisting of $d\sqrt{\CollCost}\ln(\currWin)$ slots; thus, the doubling phases contribute $O(j^2\sqrt{\CollCost}\ln(\currWin))$ slots.  By the above:
\begin{align*}
 B_j &=  O\left(j\,\sqrt{\CollCost}\ln(\currWin) \lg(w_0) + j^2\sqrt{\CollCost}\ln(\currWin)\right)\cdot \CollCost\\
&= O\left(j^2\,\sqrt{\CollCost}\ln(\currWin)\lg(w_0) \right)\cdot \CollCost\\
   &= O\left(j^3C^{(3/2)+2\epsilon}\right)
\end{align*}
\noindent since $\log(w_0) = \log(2^{\CollCost^\epsilon}) = O(\CollCost^{\epsilon})$, and $\ln(\currWin) \leq \ln(j w_0) = O(j \CollCost^{\epsilon})$.

By Lemma \ref{lem:success-going-down}, $q= 1/\eta^{c}$, where recall that $\eta=\max\{n, \CollCost\}$ and we can make the constant $c>0$ as large as we wish.  Thus,  
\begin{align*}
    \sum_{j=1}^{\infty} q^{j} \cdot B_j & \leq \sum_{j=1}^{\infty} (1/\eta^{c})^j \cdot O(j^3 C^{(3/2)+2\epsilon})\\
    & = O(1)
\end{align*}
\noindent where the last line follows for sufficiently large $c$.\medskip

{\it {\underline{Upper Bounding $G$}}.} To upper bound $G$, for each slot $i$ in the first halving phase of the first iteration, let $X_i$ be an independent random variable where $X_i = 1$ if the $i$-th slot contains a collision;  otherwise, $X_i = 0$. Recall that the windows of this first halving phase are repeatedly halved. In a window of size $w=w_0/2^k$ where $k\geq 0$ is a positive integer, by Lemma \ref{lem:upper-prob-coll}, with $m=n$ and $p=1/w$ (since $w \geq n\sqrt{\CollCost}$ and $\CollCost \geq 4$, we satisfy the requirement that $p < 1/m$): 
\begin{align*}
Pr(X_i=1) & \leq \frac{2(n/w)^2}{\left(1-(n/w)\right)}
 = \frac{2(n/(w_0/2^k))^2}{\left(1-(n/(w_0/2^k))\right)}.
\end{align*}
\noindent The expected cost due to collisions that occur in this particular  window is thus at most:
\begin{align*}
  Pr(X_i=1) \cdot d\sqrt{\CollCost}\ln(\currWin) \cdot \CollCost 
  & \leq \frac{2(n/(w_0/2^k))^2}{\left(1-(n/(w_0/2^k))\right)} d\CollCost^{(3/2)+\epsilon}.  
\end{align*}

Note that the number of windows executed starting from an initial size $w_0=2^{\CollCost^{\tunableParam}}$ down to size $n\sqrt{\CollCost}$, is at most $C^{\tunableParam} - \lg{n} -  (\lg \CollCost)/2$; for ease of presentation, let $x=C^{\tunableParam} - \lg{n} -  (\lg \CollCost)/2$. We have:
\begin{align}
  G &\leq  \sum_{k=0}^{x} \hspace{-0pt}\frac{2\left(\frac{n}{w_0/2^k}\right)^2}{\left(1-\left(\frac{n}{w_0/2^k}\right)\right)} \left(d\CollCost^{(3/2)+\epsilon}\right) \nonumber\\
  &=  \left(d\CollCost^{(3/2)+\epsilon}\right) \sum_{k=0}^{x} \hspace{-0pt}\frac{2n^2 \left(\frac{2^k}{2^{\CollCost^\tunableParam}}\right)^2}{\left(1-\frac{n 2^k}{2^{\CollCost^\tunableParam}}\right)}\nonumber\\
    &\leq    \left(d\CollCost^{(3/2)+\epsilon}\right) \frac{2n^2}{({2^{\CollCost^\tunableParam}})^2}\sum_{k=0}^{x} \hspace{-0pt}\frac{2^{2k}}{\left(1-\frac{1}{\sqrt{{\CollCost}}}\right)}\label{eq:case1}\\
    &\leq  \left(d\CollCost^{(3/2)+\epsilon}\right) \frac{4n^2}{({2^{\CollCost^\tunableParam}})^2}\sum_{k=0}^{x} \hspace{-0pt}2^{2k}\label{eq:Catleast4}\\ 
    & =  \left(d\CollCost^{(3/2)+\epsilon}\right) \frac{4n^2}{({2^{\CollCost^\tunableParam}})^2} \frac{4^{\CollCost^\tunableParam+1}}{(4-1)4^{\lg{n}} 4^{\frac{\lg{\CollCost}}{2}}}\label{eq:geometric}\\ 
   & = O\left(\CollCost^{(1/2)+\epsilon}\right).\nonumber
\end{align}
\noindent In the above, line \ref{eq:case1}  holds since  $n\sqrt{\CollCost} \leq 2^{\CollCost^{\tunableParam}}$, since we are in Case 1. Line \ref{eq:Catleast4} holds for any $\CollCost \geq 4$, and line \ref{eq:geometric} holds by the sum of the geometric series. Thus, Equation  \ref{eq:expectation} is bounded as claimed.
\end{proof}

The next lemma bounds the expected latency of our algorithm, and its argument is similar to that of Lemma \ref{lem:down_coll_cost}, although an extra $O(\CollCost^{\tunableParam})$-factor appears. Informally, this occurs since, in counting slots for the first halving phase, we have a sample size of $d\sqrt{\CollCost}\ln(\currWin)$ $\leq d\sqrt{\CollCost}\ln(w_0)$ $= d\sqrt{\CollCost} C^{\tunableParam} = O(\CollCost^{(1/2)+\tunableParam})$, and we have $\log_2 w_0=O(C^{\tunableParam})$ samples.

\begin{lemma}\label{lem:down_latency}
Assume that $n\sqrt{\CollCost} \leq w_0$. The expected latency of \mainAlgorithmAcrynom is $O(\CollCost^{(1/2)+2\tunableParam})$.
\end{lemma}

\subsection{ {\bf Case 2:} {\boldmath{$n\sqrt{\CollCost} > w_0$.}}}\label{sec:upward_direction} In this case, the good window size exceeds $w_0$, and so we argue that  \mainAlgorithmAcrynom likely succeeds during the doubling phase. A key insight is that, for Case 2, $\CollCost$ is small in terms of $n$. 

\begin{lemma}\label{l:coll-cheap-static}
Assume $n\sqrt{\CollCost} > w_0$. Then, $\CollCost = O(\log^{1/\tunableParam} n)$.
\end{lemma}
\begin{proof}
Since  $n > 2^{\CollCost^\tunableParam}/\sqrt{\CollCost}$, we have:
\begin{align}
    \lg n & > \CollCost^{\tunableParam} - (\lg \CollCost)/2\notag\\
    & >  \CollCost^{\tunableParam} -  (\CollCost^{\tunableParam}/2).\label{eqn:coll-cost-delta}
\end{align}
\noindent where the second line holds so long as $\tunableParam > \frac{\lg\lg(\CollCost)}{\lg \CollCost}$, which is true by assumption since $\tunableParam$ is a constant, and this is well-defined since $\CollCost\geq 4$.  Returning to Equation \ref{eqn:coll-cost-delta}, we have that $(2\lg n)^{1/\tunableParam} > \CollCost$, as desired, since $\tunableParam$ is a constant.
\end{proof}

The next lemma establishes bounds on expected latency and total collision cost, and the intuition is as follows. In Case 2, collisions are ``cheap'' by Lemma \ref{l:coll-cheap-static}, in the sense that $\CollCost = O(\log^{1/\tunableParam} n)$. Thus, the cost of collisions incurred until the good window is reached is small; it takes only $O(\log n)$ doublings in the doubling phase. 
After this  point the probability of success is high, and so subsequent iterations contribute little to the expected total cost of collisions. 

\begin{lemma}\label{lem:upward_latency}
Assume $n\sqrt{\CollCost} > w_0$.  Then, the expected latency is $O\left(\log^{1/(2\tunableParam)+5}n)\right)$.
\end{lemma}
\begin{proof}
The number of slots incurred for each individual halving phase is the sample size times the number of halvings, which is:  
\begin{align*}
O\left( \sqrt{\CollCost}\ln(\currWin) \right) \lg(w_0) &= O\left(\sqrt{\CollCost} \log(w_0)\log(w_0)\right)\\
&= O\left(\sqrt{\CollCost}\CollCost^{2\tunableParam}  \right)\\
&= O\left(  \log^{1/(2\epsilon) + 2}n\right)
\end{align*}
\noindent where the second line follows from $w_0=2^{\CollCost^{\tunableParam}}$, and the last line follows by applying Lemma \ref{l:coll-cheap-static}. Thus, over iterations $1$ to $j$, the number of slots contributed by the halving phases is $O(j^2\log^{ 1/(2\epsilon)+2}n)$

The number of slots in the doubling phase of the $j$-th iteration is at most:
\begin{align*}
j d\sqrt{\CollCost}\ln(\currWin) &= O\left(j \sqrt{\CollCost}\log(2^j w_0)\right) \\
&=O\left(j \sqrt{\CollCost}\, j \CollCost^{\tunableParam} \right) \\
& = O\left(j^2\, (\log^{1/(2\tunableParam)+1}n)\right) 
\end{align*}
slots, where the last line follow by Lemma \ref{l:coll-cheap-static}. By the above, summing up the number of slots executed in the doubling phases over iterations 1 to $j$ is $O(j^3\, (\log^{1/(2\tunableParam)+2}n))$.

Pessimistically, let us assume no success ever occurs during the halving phase. Focusing on the doubling phase, a good window is reached within at most $\log_2(n\sqrt{\CollCost}) + O(1) = \log_2 n + O(\log\log n)$ doublings of the window size. Since $j$ counts the number of doublings we perform in the doubling phase (recall Section \ref{sec:our_algo_static}), we reach a good window by iteration $j=c'\log n$, where $c'>0$ is a sufficiently large constant. By the above, the total number of slots up until a good window is reached is:
\begin{align*}
    O\left(j^3\, (\log^{1/(2\tunableParam)+2}n)\right) & = O\left( (\log^3 n) (\log^{1/(2\tunableParam)+2}n)\right)\\
    &= O\left(\log^{1/(2\tunableParam)+5}n\right)
\end{align*}

After this point the probability of success is high, by Lemma \ref{lem:success-going-down}, and so subsequent iterations contribute little to the expected total cost of collisions. The expected latency is:
 \begin{align*}
           &  O\left(\log^{1/(2\tunableParam)+5}n\right) + \hspace{-8pt} \sum_{j=c'\lg\lg n}^{\infty} \hspace{-8pt}O\left(\frac{1}{\eta}\right)^{j-c'\lg\lg(n) + 1} \hspace{-5pt}  O\left(j^3\, (\log^{1/(2\tunableParam)+2}n)\right) \\
           & = O\left(\log^{1/(2\tunableParam)+5}n\right)
    \end{align*}
    which completes the argument.
\end{proof}

The following lemma is proved similarly to Lemma \ref{lem:upward_latency}. The main difference is that, even if every slot contains a collision, the cost increases only by a $\CollCost = O(\log^{1/\epsilon}n)$ factor.

\begin{lemma}\label{lem:upward_collisions}
Assume $n\sqrt{\CollCost} > w_0$.  Then, the expected total cost of collisions is $O\left(\log^{3/(2\tunableParam)+5}n\right)$.  
\end{lemma}


\subsection{\bf Performance of \mainAlgorithmAcrynom}

We can now prove our upper bound.  
\medskip

\noindent\textsc{Theorem} \ref{thm:static_upper}. 
{\it \mainAlgorithmAcrynom solves the wakeup problem with the following costs:}\smallskip
\begin{enumerate}[leftmargin=15pt]
     \item {\it If\hspace{2pt} $\CollCost = \Omega(\log^{\Theta(1/\tunableParam)}n)$, then the expected latency is $O(\CollCost^{\frac{1}{2}+2\tunableParam})$ and  the expected total cost of collisions is  $O(\CollCost^{\frac{1}{2}+\tunableParam})$. }
     \item {\it Else, the expected latency is $O(\log^{5+\frac{1}{2\tunableParam}} n)$  and the expected total cost of collisions is $O( \log^{5 + \frac{3}{2\tunableParam}}n)$.
     }
 \end{enumerate} 
\begin{proof}
The relationship between $n$ and $\CollCost$ is established by the adversary (recall Section \ref{sec:model}), and there are two cases to consider. In Case 1, $n\sqrt{\CollCost} \leq w_0$. By Lemmas \ref{lem:down_coll_cost} and \ref{lem:down_latency}, the expected latency and the expected total collision cost for \mainAlgorithmAcrynom are respectively  $O(\CollCost^{(1/2)+2\tunableParam})$ and $O(\CollCost^{(1/2)+\tunableParam})$, yielding (1).

In the second case, $n\sqrt{\CollCost} > w_0$. By Lemmas, \ref{lem:upward_latency} and \ref{lem:upward_collisions},  the expected latency and expected total cost of collisions are $O(\log^{5+\frac{1}{2\tunableParam}} n)$ and $O( \log^{5 + \frac{3}{2\tunableParam}}n)$, respectively,  yielding (2).
\end{proof}

\section{Lower Bound}\label{sec:lower-bound-static}

Here, we show a trade-off between latency and the expected total cost of collisions for fair algorithms. For any fixed packet $u$, let {\boldmath{$p(t)$}} be the probability of sending in slot $t$. For any fixed algorithm $\mathcal{A}$, the sum of the sending probabilities of all packets over all time slots during the execution of $\mathcal{A}$ is the \defn{total contention}. 

\begin{lemma}\label{l:fair-algo-succeeds}
    Any fair algorithm for the wakeup problem that succeeds \whp in $\eta$ must have total contention at least $2$.
\end{lemma}

\noindent In the context of fair algorithms, the contention at any fixed slot $t$  is \defn{\con}$=\sum_{\mbox{\footnotesize{all packets}}} p(t)$.

\begin{lemma}[Biswas et al. \cite{biswas2024softeningimpactcollisionscontention}, Lemma 24]\label{lem:high-contention}
Fix any slot $t$ in any fair algorithm for the wakeup problem. If $\con> 2$, then the probability of a collision is at least $1/10$.
\end{lemma}

By Lemma \ref{lem:high-contention},  any algorithm that has a slot with contention higher than $2$ has an expected collision cost that is $\Omega(\CollCost)$. Therefore, we need only to consider algorithms where every slot in the execution has contention at most $2$. Under these circumstances, we provide a lemma that lower bounds the probability of a collision as a function of contention. We use $\conSquare$ to denote $(\con)^2$.

\begin{lemma}\label{lem:lower-bound-coll-prob}
Fix any slot $t$ where there are at least $2$ packets and $\con\leq 2$ under the execution of any fair algorithm for the wakeup problem. The probability of a collision in $t$ is   $\Omega(\conSquare)$. 
\end{lemma}

\noindent We can derive a relationship between latency and expected collision cost for fair algorithms.\smallskip

\noindent{ \textsc{Theorem} \ref{thm:static_lower}.} {\it
Let $\alpha$ and $\beta$ be constants such that $1/(2\tunableParam)\geq \alpha \geq 0$ and $\beta\geq 0$. Any fair algorithm for the wakeup problem that  has $O(\CollCost^{(1/2)+\alpha\tunableParam} \log^{\beta}n)$ latency has an expected total cost of collisions $\Omega(\CollCost^{(1/2)-\alpha\tunableParam}/\log^{\beta}n)$. } 
\begin{proof}
Let $\mathcal{A}$ be an algorithm that \whp in $\eta$ solves wakeup with $L\leq  d'\CollCost^{(1/2)+\alpha\tunableParam}\log^{\beta}n$ latency, where $d'>0$ is some constant. By Lemma \ref{l:fair-algo-succeeds}, $\mathcal{A}$ has total contention:
 \begin{align}
     \sum_{t=1}^L\con\geq 2 \label{eq:contention-sum}
 \end{align}
As we discussed, Lemma  \ref{lem:high-contention} implies that we need only consider that $\mathcal{A}$ has contention at most $2$ in every slot. Thus, we can  use Lemma \ref{lem:lower-bound-coll-prob} (noting that $n\geq 2$), which implies that the expected cost of collisions in a single slot $t$ is:
$$Pr(\text{collision in slot $t$})\,\CollCost=k\,\conSquare\,\CollCost$$ 
\noindent for some constant $k>0$. So, the expected total cost of collisions  is: 
\begin{align*}
   \sum_{t=1}^L k\, \conSquare\, \CollCost.
\end{align*}
\noindent Using Jensen's inequality for convex functions, we have:
\begin{align*}
    \CollCost k \sum_{t=1}^L \conSquare & \geq \CollCost k \frac{\left(\sum_{t=1}^L \con\right)^2}{L} \label{eq:jensen}\\
    & \geq \frac{\CollCost \,k\, 2^2 }{L}\\
    & =\Omega(\CollCost^{(1/2)-\alpha\tunableParam}/\log^{\beta}n) 
\end{align*}
 \noindent where the second line follows from Equation \ref{eq:contention-sum}, and the last line follows since $L\leq  d'\CollCost^{(1/2)+\alpha\tunableParam}\log^{\beta}n$.
\end{proof}

As evidence that our result is not far off from the asymptotic optimal, we note that any fair algorithms with latency $O(\CollCost^{\frac{1}{2}+2\tunableParam})$ must have an expected total cost of collisions that is $\Omega(\CollCost^{\frac{1}{2}-2\tunableParam})$, which is an  $O(\CollCost^{3\tunableParam})$-factor smaller than that achieved by \mainAlgorithmAcrynom.  \medskip

\section{Conclusion and Future Work}

Our work gives a preliminary theoretical exploration  addressing collision costs. For future work, we aim to prove more general lower bounds.  We also wish to derive results for the dynamic setting, where packets arrive over time rather than all being initially present.\medskip

\noindent{\bf Acknowledgements.} This research is supported by the National Science Foundation grants CNS-2210300 and CCF-2144410.



\begin{thebibliography}{20}


\ifx \showCODEN    \undefined \def \showCODEN     #1{\unskip}     \fi
\ifx \showISBNx    \undefined \def \showISBNx     #1{\unskip}     \fi
\ifx \showISBNxiii \undefined \def \showISBNxiii  #1{\unskip}     \fi
\ifx \showISSN     \undefined \def \showISSN      #1{\unskip}     \fi
\ifx \showLCCN     \undefined \def \showLCCN      #1{\unskip}     \fi
\ifx \shownote     \undefined \def \shownote      #1{#1}          \fi
\ifx \showarticletitle \undefined \def \showarticletitle #1{#1}   \fi
\ifx \showURL      \undefined \def \showURL       {\relax}        \fi
\providecommand\bibfield[2]{#2}
\providecommand\bibinfo[2]{#2}
\providecommand\natexlab[1]{#1}
\providecommand\showeprint[2][]{arXiv:#2}

\bibitem[Anderton et~al\mbox{.}(2021)]%
        {anderton:windowed}
\bibfield{author}{\bibinfo{person}{William~C. Anderton}, \bibinfo{person}{Trisha Chakraborty}, {and} \bibinfo{person}{Maxwell Young}.} \bibinfo{year}{2021}\natexlab{}.
\newblock \showarticletitle{Windowed Backoff Algorithms for {WiFi}: Theory and Performance under Batched Arrivals}.
\newblock \bibinfo{journal}{\emph{Distributed Computing}}  \bibinfo{volume}{34} (\bibinfo{year}{2021}), \bibinfo{pages}{367--393}.
\newblock


\bibitem[Banicescu et~al\mbox{.}(2024)]%
        {banicescu2024survey}
\bibfield{author}{\bibinfo{person}{Ioana Banicescu}, \bibinfo{person}{Trisha Chakraborty}, \bibinfo{person}{Seth Gilbert}, {and} \bibinfo{person}{Maxwell Young}.} \bibinfo{year}{2024}\natexlab{}.
\newblock \showarticletitle{A Survey on Adversarial Contention Resolution}.
\newblock \bibinfo{journal}{\emph{Comput. Surveys}} (\bibinfo{year}{2024}).
\newblock


\bibitem[Bar-Yehuda et~al\mbox{.}(1992)]%
        {bar1992time}
\bibfield{author}{\bibinfo{person}{Reuven Bar-Yehuda}, \bibinfo{person}{Oded Goldreich}, {and} \bibinfo{person}{Alon Itai}.} \bibinfo{year}{1992}\natexlab{}.
\newblock \showarticletitle{On the time-complexity of broadcast in multi-hop radio networks: An exponential gap between determinism and randomization}.
\newblock \bibinfo{journal}{\emph{J. Comput. System Sci.}} \bibinfo{volume}{45}, \bibinfo{number}{1} (\bibinfo{year}{1992}), \bibinfo{pages}{104--126}.
\newblock


\bibitem[Biswas et~al\mbox{.}(2024a)]%
        {biswas2024softening}
\bibfield{author}{\bibinfo{person}{Umesh Biswas}, \bibinfo{person}{Trisha Chakraborty}, {and} \bibinfo{person}{Maxwell Young}.} \bibinfo{year}{2024}\natexlab{a}.
\newblock \showarticletitle{Softening the Impact of Collisions in Contention Resolution}. In \bibinfo{booktitle}{\emph{International Symposium on Stabilizing, Safety, and Security of Distributed Systems}}. \bibinfo{pages}{398--416}.
\newblock


\bibitem[Biswas et~al\mbox{.}(2024b)]%
        {biswas2024softeningimpactcollisionscontention}
\bibfield{author}{\bibinfo{person}{Umesh Biswas}, \bibinfo{person}{Trisha Chakraborty}, {and} \bibinfo{person}{Maxwell Young}.} \bibinfo{year}{2024}\natexlab{b}.
\newblock \bibinfo{title}{Softening the Impact of Collisions in Contention Resolution}.
\newblock
\showeprint[arxiv]{2408.11275}~[cs.DC]
\urldef\tempurl%
\url{https://arxiv.org/abs/2408.11275}
\showURL{%
\tempurl}


\bibitem[Chang et~al\mbox{.}(2019)]%
        {chang:exponential-jacm}
\bibfield{author}{\bibinfo{person}{Yi-Jun Chang}, \bibinfo{person}{Tsvi Kopelowitz}, \bibinfo{person}{Seth Pettie}, \bibinfo{person}{Ruosong Wang}, {and} \bibinfo{person}{Wei Zhan}.} \bibinfo{year}{2019}\natexlab{}.
\newblock \showarticletitle{Exponential Separations in the Energy Complexity of Leader Election}.
\newblock \bibinfo{journal}{\emph{ACM Transactions on Algorithms}} \bibinfo{volume}{15}, \bibinfo{number}{4}, Article \bibinfo{articleno}{49} (\bibinfo{date}{Oct.} \bibinfo{year}{2019}), \bibinfo{numpages}{31}~pages.
\newblock


\bibitem[Farach-Colton et~al\mbox{.}(2006)]%
        {farach-colton:lower-bounds}
\bibfield{author}{\bibinfo{person}{Mart{\'i}n Farach-Colton}, \bibinfo{person}{Rohan~J. Fernandes}, {and} \bibinfo{person}{Miguel~A. Mosteiro}.} \bibinfo{year}{2006}\natexlab{}.
\newblock \showarticletitle{Lower Bounds for Clear Transmissions in Radio Networks}. In \bibinfo{booktitle}{\emph{Proceedings of the Latin American Symposium on Theoretical Informatics (LATIN)}}. \bibinfo{pages}{447--454}.
\newblock


\bibitem[Fineman et~al\mbox{.}(2016a)]%
        {fineman:contention}
\bibfield{author}{\bibinfo{person}{Jeremy~T. Fineman}, \bibinfo{person}{Seth Gilbert}, \bibinfo{person}{Fabian Kuhn}, {and} \bibinfo{person}{Calvin Newport}.} \bibinfo{year}{2016}\natexlab{a}.
\newblock \showarticletitle{Contention Resolution on a Fading Channel}. In \bibinfo{booktitle}{\emph{Proceedings of the ACM Symposium on Principles of Distributed Computing (PODC)}}. \bibinfo{pages}{155--164}.
\newblock


\bibitem[Fineman et~al\mbox{.}(2016b)]%
        {fineman:contention2}
\bibfield{author}{\bibinfo{person}{Jeremy~T. Fineman}, \bibinfo{person}{Calvin Newport}, {and} \bibinfo{person}{Tonghe Wang}.} \bibinfo{year}{2016}\natexlab{b}.
\newblock \showarticletitle{Contention Resolution on Multiple Channels with Collision Detection}. In \bibinfo{booktitle}{\emph{Proceedings of the ACM Symposium on Principles of Distributed Computing (PODC)}}. \bibinfo{pages}{175--184}.
\newblock


\bibitem[Ger{\'{e}}b{-}Graus and Tsantilas(1992)]%
        {Gereb-GrausT92}
\bibfield{author}{\bibinfo{person}{Mih{\'{a}}ly Ger{\'{e}}b{-}Graus} {and} \bibinfo{person}{Thanasis Tsantilas}.} \bibinfo{year}{1992}\natexlab{}.
\newblock \showarticletitle{Efficient Optical Communication in Parallel Computers}. In \bibinfo{booktitle}{\emph{Proceedings of the Annual ACM Symposium on Parallel Algorithms and Architectures (SPAA)}}. \bibinfo{pages}{41--48}.
\newblock


\bibitem[Gilbert et~al\mbox{.}(2021)]%
        {gilbert2021contention}
\bibfield{author}{\bibinfo{person}{Seth Gilbert}, \bibinfo{person}{Calvin Newport}, \bibinfo{person}{Nitin Vaidya}, {and} \bibinfo{person}{Alex Weaver}.} \bibinfo{year}{2021}\natexlab{}.
\newblock \showarticletitle{Contention resolution with predictions}. In \bibinfo{booktitle}{\emph{Proceedings of the ACM Symposium on Principles of Distributed Computing (PODC)}}. \bibinfo{pages}{127--137}.
\newblock


\bibitem[Greenberg and Leiserson(1985)]%
        {GreenbergL85}
\bibfield{author}{\bibinfo{person}{Ronald~I. Greenberg} {and} \bibinfo{person}{Charles~E. Leiserson}.} \bibinfo{year}{1985}\natexlab{}.
\newblock \showarticletitle{Randomized Routing on Fat-Trees}. In \bibinfo{booktitle}{\emph{Proceedings of the IEEE Symposium on Foundations of Computer Science (FOCS)}}. \bibinfo{pages}{241--249}.
\newblock


\bibitem[Jurdzi{\'n}ski and Stachowiak(2002)]%
        {jurdzinski2002probabilistic}
\bibfield{author}{\bibinfo{person}{Tomasz Jurdzi{\'n}ski} {and} \bibinfo{person}{Grzegorz Stachowiak}.} \bibinfo{year}{2002}\natexlab{}.
\newblock \showarticletitle{Probabilistic algorithms for the wakeup problem in single-hop radio networks}. In \bibinfo{booktitle}{\emph{Proceedings of the International Symposium on Algorithms and Computation}}. \bibinfo{pages}{535--549}.
\newblock


\bibitem[Kushilevitz and Mansour(1993)]%
        {kushilevitz1993omega}
\bibfield{author}{\bibinfo{person}{Eyal Kushilevitz} {and} \bibinfo{person}{Yishay Mansour}.} \bibinfo{year}{1993}\natexlab{}.
\newblock \showarticletitle{An {$\Omega$(D$\log$(N/D))} lower bound for broadcast in radio networks}. In \bibinfo{booktitle}{\emph{Proceedings of the Annual ACM Symposium on Principles of Distributed Computing (PODC)}}. \bibinfo{pages}{65--74}.
\newblock


\bibitem[Kushilevitz and Mansour(1998)]%
        {kushilevitz1998omega}
\bibfield{author}{\bibinfo{person}{Eyal Kushilevitz} {and} \bibinfo{person}{Yishay Mansour}.} \bibinfo{year}{1998}\natexlab{}.
\newblock \showarticletitle{An {$\Omega$(D$\log$(N/D))} lower bound for broadcast in radio networks}.
\newblock \bibinfo{journal}{\emph{SIAM journal on Computing}} \bibinfo{volume}{27}, \bibinfo{number}{3} (\bibinfo{year}{1998}), \bibinfo{pages}{702--712}.
\newblock


\bibitem[Moscibroda and Wattenhofer(2006)]%
        {moscibroda2006complexity}
\bibfield{author}{\bibinfo{person}{Thomas Moscibroda} {and} \bibinfo{person}{Roger Wattenhofer}.} \bibinfo{year}{2006}\natexlab{}.
\newblock \showarticletitle{The Complexity of Connectivity in Wireless Networks.}. In \bibinfo{booktitle}{\emph{Proceedings of the IEEE Conference on Computer Communications (INFOCOM)}}. \bibinfo{pages}{1--13}.
\newblock


\bibitem[Newport(2014a)]%
        {newport2014radio}
\bibfield{author}{\bibinfo{person}{Calvin Newport}.} \bibinfo{year}{2014}\natexlab{a}.
\newblock \showarticletitle{Radio network lower bounds made easy}. In \bibinfo{booktitle}{\emph{Proceedings of the International Symposium on Distributed Computing (DISC)}}. \bibinfo{pages}{258--272}.
\newblock


\bibitem[Newport(2014b)]%
        {newport:radio-journal}
\bibfield{author}{\bibinfo{person}{Calvin Newport}.} \bibinfo{year}{2014}\natexlab{b}.
\newblock \showarticletitle{Radio Network Lower Bounds Made Easy}. In \bibinfo{booktitle}{\emph{Distributed Computing}}, \bibfield{editor}{\bibinfo{person}{Fabian Kuhn}} (Ed.). \bibinfo{pages}{258--272}.
\newblock


\bibitem[Willard(1986)]%
        {willard:loglog}
\bibfield{author}{\bibinfo{person}{Dan~E. Willard}.} \bibinfo{year}{1986}\natexlab{}.
\newblock \showarticletitle{Log-logarithmic Selection Resolution Protocols in a Multiple Access Channel}.
\newblock \bibinfo{journal}{\emph{SIAM J. Comput.}} \bibinfo{volume}{15}, \bibinfo{number}{2} (\bibinfo{date}{May} \bibinfo{year}{1986}), \bibinfo{pages}{468--477}.
\newblock


\bibitem[Zhang(2006)]%
        {zhang2006routing}
\bibfield{author}{\bibinfo{person}{Zhensheng Zhang}.} \bibinfo{year}{2006}\natexlab{}.
\newblock \showarticletitle{Routing in intermittently connected mobile ad hoc networks and delay tolerant networks: overview and challenges}.
\newblock \bibinfo{journal}{\emph{IEEE Communications Surveys \& Tutorials}} \bibinfo{volume}{8}, \bibinfo{number}{1} (\bibinfo{year}{2006}), \bibinfo{pages}{24--37}.
\newblock


\end{thebibliography}


\end{document}